\documentclass[12pt, hidelinks]{article}
\usepackage[english]{babel}
\usepackage{amssymb,amsmath,amsthm,amsfonts,mathtools}
\usepackage[top=2.8cm, bottom=2.8cm, left=2.8cm, right=2.8cm]{geometry}
\usepackage{graphicx}
\usepackage{amsfonts}
\usepackage[utf8]{inputenc}
\usepackage{tikz}
\usepackage{tikz-cd}
\usetikzlibrary{positioning}
\usetikzlibrary{trees}
\usetikzlibrary{arrows}
\usetikzlibrary{shapes}
\usepackage{verbatim}
\usepackage{appendix}
\usepackage[hyphens]{url}
\usepackage{hyperref}
\hypersetup{breaklinks=true}
\usepackage{datetime}
\usepackage{setspace}
\usepackage{bookmark}
\usepackage{quoting} 
\usepackage{enumerate}
\usepackage{floatrow}

\usepackage[authordate, backend=biber, natbib=true, maxcitenames=2, hyperref, backref, backrefstyle=none]{biblatex-chicago}

\DefineBibliographyStrings{english}{%
  backrefpage = {Cited on p.},
  backrefpages = {Cited on p.},
}

\usepackage{xcolor}
\usepackage{pgfplots}
\pgfplotsset{width=7cm,compat=1.8}
\usepackage{pgfplotstable}

\newtheorem{theorem}{Theorem}
\newtheorem{definition}{Definition}
\newtheorem{remark}{Remark}

\title{Rationality and correctness in $n$-player games\footnote{The authors conducted this study without external funding and have no financial or non-financial conflicts of interest to disclose.}}

\author{Lorenzo Bastianello\footnote{Ca' Foscari University of Venice, Venice, Italy. e-mail: lorenzo.bastianello@unive.it} \and
Mehmet S. Ismail\footnote{Department of Political Economy, King's College London, London, UK. e-mail: mehmet.ismail@kcl.ac.uk.}}

\date{Revised: 11 December 2023 \\ First version: 20 September 2022}

\begin{document}
\maketitle

\begin{abstract}
There are two well-known sufficient conditions for Nash equilibrium in two-player games: mutual knowledge of rationality (MKR) and mutual knowledge of conjectures. MKR assumes that the \textit{concept} of rationality is mutually known. In contrast, mutual knowledge of conjectures assumes that a given profile of conjectures is mutually known, which has long been recognized as a strong assumption. In this note, we introduce a notion of ``mutual assumption of rationality and correctness'' (MARC), which conceptually aligns more closely with the MKR assumption. We present two main results. Our first result establishes that MARC holds in every two-person zero-sum game. In our second theorem, we show that MARC does not in general hold in $n$-player games. 
\end{abstract}


\newpage
\section{Introduction and related literature}

\begin{quote} ``Suppose that the game being
played (i.e., both payoff functions), the rationality of the players, and their conjectures are all mutually known. Then the conjectures constitute a Nash equilibrium.'' \citep[p.\ 1161]{aumann1995}
\end{quote}

Understanding the epistemic conditions that give rise to  a Nash equilibrium  has been a topic of interest in both economics and computer science, among other fields. Two well-known sufficient conditions for Nash equilibrium in two-player games are mutual knowledge of rationality (MKR) and mutual knowledge of conjectures. In their seminal paper, \citet{aumann1995} prove that MKR and mutual knowledge of conjectures are sufficient for conjectures to form a Nash equilibrium (see the above quote).\footnote{\citet{aumann1987} does not assume common knowledge of conjectures (i.e., beliefs) but assumes common prior from which these conjectures are derived, and he shows that these assumptions imply correlated equilibrium behavior; for further discussion, see, e.g., \citet{battigalli1999}.} Other notable works in this area include those of \citet{brandenburger1987}, \citet{tan1988}, and a comprehensive review of this literature can be found in works by \citet{perea2012} and \citet{dekel2015}. 

The assumptions of MKR and mutual knowledge of conjectures serve different roles and should be distinguished. MKR assumes that the \textit{concept} of rationality is mutually known, without assigning some rational strategies to players. Accordingly, strategies consistent with MKR are determined by e.g. eliminating strictly dominated strategies.  In contrast, mutual knowledge of conjectures requires that a given profile of conjectures is mutually known, which has long been recognized as a strong assumption. Several prominent works by authors, including \citet{aumann1995}, \citet{bach2014}, and \citeauthor{perea2024} (forthcoming, Chapter 4), acknowledged the issues related to the assumption of mutual knowledge of conjectures:
\begin{quote} ``Both of our theorems assume some form of knowledge (either mutual or common) of the players' conjectures. Although the assumption of knowledge regarding others' actions is undoubtedly strong, there are circumstances in which it could hold.'' (ibid., p.\ 1176)
\end{quote}
\begin{quote} ``The conceptual issue imposed by assuming belief in opponents' conjectures persists.'' \citep[p.\ 57]{bach2014}
\end{quote}

In this paper, we introduce a notion of ``knowledge of correctness'' which conceptually aligns more closely with the MKR assumption. To proceed, we first provide informal definitions for the concepts of rationality and correctness. A player is rational if she chooses a strategy that maximizes her expected payoff given her conjecture about other players' strategies.  A player's conjecture  is correct if her conjecture coincides with the actual (ex-post) choice of other players. We say that mutual assumption of rationality and correctness (MARC) is satisfied if (i) each player $i$ is both rational and correct, and (ii) each player $i$ is rational, assuming that each player $j\neq i$ is both rational and correct. 

For interpreting our framework, we suggest keeping the following experimental setup in mind. Using a computer, each player $i$ inputs a mixed strategy, which leads to an expected payoff derived from the resulting mixed strategy profile. Players also input their conjectures, though these are not payoff-relevant. After the game, we check whether the conjectures of the players were correct. 

We present two main results in this paper. First, we prove that MARC holds in two-person zero-sum games, as stated in Theorem~\ref{thm:zerosum}. This consistency result can be attributed to the properties inherent in maximin strategies. Specifically, in zero-sum settings, each player's maximin strategy takes into account the possibility that her opponent will correctly anticipate her choices and respond optimally. Furthermore, a profile of maximin strategies form a Nash equilibrium, and hence MARC holds in this class of games. Second, Theorem~\ref{thm:main} establishes that MARC does not in general hold in $n$-player games. To elaborate, for every $n\geq 2$, we construct an $n$-player game such that assuming MARC leads to a contradiction. 

In Remark~\ref{rem:Nash}, we revisit  Aumann and Brandenburger's (\citeyear{aumann1995}) preliminary observation on sufficient conditions for Nash equilibrium in our context. These conditions continue the tradition established by earlier works, including by \citet{brandenburger1987}, \citet{tan1988},   \citet{perea2007}, \citet{barelli2009}, \citet{perea2012}, \citet{bach2014}, and \citet{bach2020}. Each of these works assumes a form of correctness of conjectures (for instance, mutual knowledge of conjectures in  Aumann and Brandenburger \citeyear{aumann1995}, and simple belief hierarchy in Perea \citeyear{perea2012}) to obtain Nash equilibrium. The first significant difference between our framework and these frameworks is that our model studies whether the actual choices of players constitute a Nash equilibrium. In contrast, the earlier models focus on whether the beliefs, rather than the actual choices, form a Nash equilibrium. In addition, our model allows for experimental falsification to verify the correctness of a conjecture.

\citet{perea2007} provide a set of sufficient conditions for choosing a Nash strategy, a strategy that is best response for a player given the Nash equilibrium strategies of other players, from the perspective of the beliefs of one player. \citet{polak1999} shows that Aumann and Brandenburger's (\citeyear{aumann1995}) conditions imply common knowledge of rationality in complete information games.  \citet{barelli2009} weakens earlier the epistemic conditions provided by \citet{aumann1995} without requiring common knowledge of rationality. \citet{bach2014} further weaken the earlier conditions by assuming these conditions on some pairs of players rather than assuming them for all pairs of players to obtain Nash equilibrium. Like \citet{barelli2009}, \citet{bach2014} emphasize that a common knowledge of rationality is not necessary for Nash equilibrium.

The implications of knowledge structures in games have received extensive attention over several decades, involving contributions from computer scientists, logicians, and economists. Beyond the works already cited, other notable research in this area includes studies by 
\citet{harsanyi1965,bicchieri1988,gilboa1990,morris1995,gul1998,aumann1998,halpern2002}. Our work is also related to the literature on conditional commitments \citet{howard1971,tennenholtz2004,kalai2010,vanderhoek2013,oesterheld2019,bastianello2022}, Stackelberg games and `magical thinking' \citep{quattrone1984}. Our framework is distinct in part because it does not incorporate assumptions such as magical thinking or conditional commitments. Instead, our aim is to characterize the class of games by assuming both the rationality and correctness of players. As previously mentioned, players' conjectures can be either correct or incorrect, and this can be tested in the lab. Specifically, in $n$-person games, we show that MARC does not hold, implying that, under the assumption of rationality, some players must be incorrect. However, in zero-sum games, we find that MARC is indeed satisfied.

\section{The setup}
\label{sec:setup}

We define $N = \{1, \ldots, n\}$ as the finite set of players, and $A_i$ as the finite set of pure actions available to player $i$. The set $\Delta A_i$ represents all probability distributions over $A_i$. We further define $\Delta A = \bigtimes_{i \in N} \Delta A_i$ and $\Delta A_{-i} = \bigtimes_{j \in N \setminus \{i\}} \Delta A_j$. To denote a \textbf{mixed strategy} of player $i$, we use the notation  $t_i \in \Delta A_i$. The notation $t\in \Delta A$ represents a mixed strategy profile, and $t_{-i} \in \Delta A_{-i}$ represents a mixed strategy profile for all players except $i$. 

We let $c^i_j \in \Delta A_j$ denote a \textbf{conjecture} (i.e., belief) of player $i \neq j$ regarding player $j$'s mixed strategy choice, and $c^i_{-i} \in \Delta A_{-i}$ denote player $i$'s conjecture about the other players' mixed strategy choices. 

Suppose that player $i$ chooses the mixed strategy $s_i\in \Delta A_{-i}$. Then, we call $s_i$ the \textbf{actual} strategy (or choice) of player $i$. Likewise, $s\in \Delta A$ denotes the actual mixed strategy profile, and $s_{-i} \in \Delta A_{-i}$ denotes the actual mixed strategy choices of all players other than $i$. 

Lastly, we define $u_i:\Delta A\rightarrow \mathbb{R}$ as the von Neumann-Morgenstern expected payoff function for player $i\in N$. We let $G=(\Delta A_i, u_i)_{i\in N}$ denote an $n$-player non-cooperative game in normal form. 

The interpretation of our results is most straightforward with the following experimental setup in mind, which can be used to test and falsify our assumptions. Imagine that each subject (player) $i$ inputs a mixed strategy $s_i$ using a computer. Each player $i$'s expected payoffs are then given by $u_i(s)$. In addition, each subject is also asked to type in her conjecture $c^i_{-i} \in \Delta A_{-i}$ about the mixed strategy choices of the other players, though the entered conjectures are not payoff-relevant. After the game, we simply check whether the conjectures of the players were correct, which are defined as follows.

\begin{definition}
\label{def:correctness}
Let $s_{-i}$ be the actual strategy profile of all players but $i$. Player $i$'s conjecture $c^i_{-i}$ is said to be \textbf{correct} if $c^i_{-i}=s_{-i}$. Player $i$ is called \textit{correct} if her conjecture is correct.
\end{definition}

In words, player $i$'s conjecture about other players' strategy is correct if the conjecture entered by $i$ matches the other players' actual mixed strategy profile.

We proceed to define the rationality notion used throughout the text.

\begin{definition}
\label{def:rationality}
Player $i$ is said to be \textit{rational} assuming that the other players choose $c^i_{-i}\in \Delta A_{-i}$, if $i$ chooses a strategy $s_i$ that solves the following maximization problem:
\begin{equation}
\label{eq:maximization_problem}
\begin{aligned}
\max_{t_{i} \in \Delta A^{i}} \quad & u_i(t_i,t_{-i} ) \\
\text{s.t.} \quad & t_{-i} = c^i_{-i}.
\end{aligned}
\end{equation}
In this context, $s_i$ is referred to as a \textbf{best response} for player $i$ to $c^i_{-i}$. Player $i$ is said to be \textbf{rational} if she is rational assuming some conjecture $c^i_{-i}\in \Delta A_{-i}$.
\end{definition}

This is a definition of Bayesian rationality \`a la Savage, where a player's actual strategy maximizes her expected utility given her conjecture about the other players' behavior. We then offer a helpful remark, serving as an analog to Aumann and Brandenburger's (1995) preliminary observation within this framework.

\begin{remark}[Aumann and Brandenburger, 1995]
    \label{rem:Nash}
    Let $s$ be the actual strategy profile played in an $n$-person game. For each player $i$, assume that $i$ is correct, $c^i_{-i}=s_{-i}$, and that  $i$ is rational assuming $c^i_{-i}$. Then, $s$ is a Nash equilibrium.
\end{remark}

\begin{proof}
For each player $i$, let $s_{-i}$ be the actual strategy profile of all players but $i$. Then, by definition, if player $i$ is correct, then $c^i_{-i}=s_{-i}$. In addition, if player $i$ is rational, then according to the maximization problem~(\ref{eq:maximization_problem}),  $s_{i}$ is a best response to $s_{-i}$ for each $i$, as desired.
\end{proof}

Next, we recall the concept of maximin strategy in zero-sum games and we underline its relation with rationality and correctness. Using the usual notation, a maximin strategy $s_i$ solves
\[
s_{i}\in \arg\max_{t_{i} \in \Delta A_{i}} \min_{t_j \in \Delta A_{j}}  u_{i}(t_{i}, t_j).
\]
For all strategies $t_i\in \Delta A_{i}$ of player $i$, player $j$ chooses $t_j \in \Delta A_{j}$ that minimizes player $i$'s utility (or, equivalently, maximizes player $j$'s own utility).  Note that this can be equivalently written as follows. 

\begin{definition}
\label{def:maximin}
Let $G$ be a zero-sum game. Player $i$'s mixed strategy $s_i$ is said to be a \textit{maximin strategy} if $s_i$ solves the following maximization problem:
\begin{equation}
\label{eq:maximin}
\begin{aligned}
\max_{t_{i} \in \Delta A_{i}} \quad & u_i(t_i, t_j) \\
\text{s.t.} \quad & t_{i} = c^j_{i},\\
& t_j\in \arg\min_{t'_{j} \in \Delta A_{j}}  u_{i}(c^j_{i}, t'_{j}) .
\end{aligned}
\end{equation}
\end{definition}

In words, a strategy $s_i$ is called a maximin strategy if player $i$ maximizes her utility assuming that player $j\neq i$ is correct (the first constraint) and player $j$ chooses a strategy that minimizes $i$'s utility  (the second constraint).  Our definition is designed to clearly emphasize two key assumptions: correctness and best-response behavior. This sets the stage for our next definition of rationality, which assumes that other players are not only correct but also maximize utility.

\begin{definition}
\label{def:rationality_assumption}
Player $i$ is said to be \textit{rational assuming that each player $j\neq i$ is both rational and correct}, if $i$ chooses a strategy $s_i$ that solves the following maximization problem:
\begin{equation}
\label{eq:maximization_problem2}
\begin{aligned}
\max_{t_{i} \in \Delta A_{i}} \quad & u_i(t_i, t_{-i}) \\
\text{s.t.} \quad & t_{-j} =c^j_{-j} \text{ for all } j\neq i,\\
& t_j\in \arg\max_{t'_{j} \in \Delta A_{j}}  u_{j}(t'_{j},c^j_{-j}) \text{ for all } j\neq i.
\end{aligned}
\end{equation}
\end{definition}

Analogous to the maximization problem in the zero-sum case,  player $i$ assumes that: (i) each player $j$ is correct, meaning that, $t_{-i}=c^j_{-j}$; and (ii) each player $j$ maximizes utility, assuming $c^j_{-j}$. These two assumptions, namely that each player $j\neq i$ is both correct and rational, and Remark~\ref{rem:Nash} lead to the following observation.

\begin{remark}
    \label{rem:Nash-others}
 Definition \ref{def:rationality_assumption} implies that in the $(n-1)$-person game induced by player $i$'s own choice $t_i$, the remaining players' strategies form a Nash equilibrium.
\end{remark}

We conclude this section with the following definition.

\begin{definition}
\label{def:mar_mac}
\textit{Mutual assumption of rationality and correctness} (\textbf{MARC}) is said to hold in an $n$-person game $G$, if the following statements are satisfied:
\begin{enumerate}
    \item Each player $i\in N$ is both rational (Definition~\ref{def:rationality}) and correct (Definition~\ref{def:correctness}). 
    \item Each player $i\in N$ is rational assuming that each player $j\neq i$ is both rational and correct (Definition~\ref{def:rationality_assumption}). 
\end{enumerate}
\end{definition}

\section{Main results}
\label{sec:results}

We first explore whether there is any important subclass of games in which MARC is satisfied. Our first main theorem establishes that MARC does not lead to any logical contradictions in zero-sum games.

\begin{theorem}[Consistency in zero-sum games]
\label{thm:zerosum}
MARC holds in every finite two-player zero-sum game.
\end{theorem}

\begin{proof}
By Definition~\ref{def:rationality_assumption}, if player $i$ is rational assuming that player $j\neq i$ is both rational and correct, then $i$ chooses a strategy $s_i$ that solves the following maximization problem:

\begin{equation*}
\begin{aligned}
\max_{t_{i} \in \Delta A_{i}} \quad & u_i(t_i, t_j) \\
\text{s.t.} \quad & t_{i} = c^j_{i},\\
& t_j\in \arg\max_{t'_{j} \in \Delta A_{j}}  u_{j}(c^j_{i}, t'_{j}) .
\end{aligned}
\end{equation*}

Since the game is of zero-sum, the maximization problem reduces to the following problem: 

\begin{equation*}
\begin{aligned}
\max_{t_{i} \in \Delta A_{i}} \quad & u_i(t_i, t_j) \\
\text{s.t.} \quad & t_{i} = c^j_{i},\\
& t_j\in \arg\min_{t'_{j} \in \Delta A_{j}}  u_{i}(c^j_{i}, t'_{j}) .
\end{aligned}
\end{equation*}
Thus, by Definition~\ref{def:maximin}, for every player $i$, $s_i$ is a maximin strategy.

We have established  that a strategy profile $s$ becomes the maximin strategy profile under two conditions: each player $i$ is rational assuming that the other player $j \neq i$ is both rational and correct. Additionally, assume that each player $i$ is both rational and correct. Then, by Remark~\ref{rem:Nash}, $s$ must be a Nash equilibrium, which is indeed the case because $s$ is a maximin strategy profile in a (finite) zero-sum game. As a result, MARC holds.\end{proof}

As mentioned in the Introduction, the underlying intuition for this theorem can be attributed to two primary implications arising from von Neumann's (\citeyear{neumann1928}) minimax theorem. The first implication is that a maximin strategy inherently accounts for the scenario where player $j \neq i$ correctly anticipates the strategy of player $i$ and responds optimally. The second implication is that each player's maximin strategy is a best response to the maximin strategy used by the other player.  Thus, MARC holds in zero-sum games.

The next theorem establishes that the previous result does not extend to $n$-person games.

\begin{theorem}[Inconsistency]
\label{thm:main}
For every $n\geq 2$, there exists an $n$-player game such that MARC does not hold.
\end{theorem}

\begin{proof}
To reach a contradiction, assume that MARC holds. An $n$-player game will be constructed for every $n\geq 2$, wherein $s_i$ contradicts player $i$'s rationality and correctness. First, consider the $2 \times 2$ game presented in Figure~\ref{fig:theorem}.

\begin{figure}[h!]
\centering
\[
	\begin{array}{ r|c|c| }
	\multicolumn{1}{r}{}
	&  \multicolumn{1}{c}{x_1}
	& \multicolumn{1}{c}{x_2}\\
	\cline{2-3}
	x_1 &  2,1 & 0,0 \\
	\cline{2-3}
	x_2 &  0,0 & 1,2 \\
	\cline{2-3}
	\end{array}
\]
\caption{A counterexample to MARC}
\label{fig:theorem}
\end{figure}

By Definition~\ref{def:rationality_assumption}, if player $i$ is rational, assuming that player $j\neq i$ is both rational and correct, then player $i$ chooses a strategy $s_i$ that solves:
\begin{equation*}
\begin{aligned}
\max_{t_{i} \in \Delta A_{i}} \quad & u_i(t_i, t_j) \\
\text{s.t.} \quad & t_{i} =c^j_{i}, \\
& t_j\in \arg\max_{t'_{j} \in \Delta A_{j}}  u_{j}(t'_{j},c^j_{i}).
\end{aligned}
\end{equation*}

For each player $i$, the pure strategy $x_i$ uniquely solves this maximization problem. In fact, if $s_i=x_i$, by correctness of player $j$ we obtain that $c^j_i=x_i$ and hence $t_j=x_i$. Similarly for $s_i=x_j$ it follows that $t_j=x_j$. Therefore, for each player $i$, $x_i$ solves the maximization problem of player $i$ and, as a result, we obtain that $(s_1, s_2)=(x_1,x_2)$. However, it is impossible for players to be both rational and correct simultaneously. For instance, player 1 may be rational if her conjecture about player 2 were $c^1_2=x_1$, yet this would render her incorrect as $s_2=x_2$. Conversely, player 1 may be correct if her conjecture is $c^1_2=x_2$, but then she would not be rational, assuming this conjecture. Thus, we obtain the desired contradiction to the assumption that MARC holds in this game.

We next generalize this game to the $n$-player case with $n>2$. Assume that for every player $j\in N$, $A_j=\{x_1,x_2\}$ and that $x_1$ is the strictly dominant strategy for every player $j'\in N\setminus \{1,2\}$.\footnote{Specific values of the payoffs do not matter.} For player 1 and player 2, the payoffs are defined as follows: for each $i\in \{1,2\}$ and each pure strategy profile $t$, $u_i(t)=u_i(t_1,t_2)$, where $u_i(t_1,t_2)$ denotes player $i$'s payoff in Figure~\ref{fig:theorem}.

Similar to the two-player case, for each player $i\in \{1,2\}$, if player $i$ is rational, assuming that each player $j\neq i$ is both rational and correct, then the pure strategy $x_i$ uniquely solves the maximization problem as outlined in  Definition~\ref{def:rationality_assumption}. Consequently, we obtain that $(s_1, s_2, ..., s_n)=(x_1,x_2,x_1,x_1,..., x_1)$. Thus, analogous to the two-player case, players cannot be both rational and correct. 
\end{proof}

The distinction between Remark~\ref{rem:Nash} and Theorem~\ref{thm:main} can be summarized as follows. In Remark~\ref{rem:Nash}, each player's given conjecture  \emph{is} correct, and each player is rational given his or her conjecture. However, no player makes any assumption regarding the correctness of the other player. In contrast, in Theorem~\ref{thm:main}, each player assumes the other player is both rational and correct.

Note that the counterexample provided in Theorem~\ref{thm:main} is not an `isolated' game. Only in very specific games, does the MARC hold in $n$-person games. While it is clear that MARC would hold in games with a strictly dominant strategy equilibrium, it is less apparent whether it holds in dominance solvable games. We present a counterexample below, showing that MARC does not hold in this more general class of games.
\[
	\begin{array}{ r|c|c| }
	\multicolumn{1}{r}{}
	&  \multicolumn{1}{c}{x_2}
	& \multicolumn{1}{c}{y_2}\\
	\cline{2-3}
	x_1 &  1,1 & 3,2 \\
	\cline{2-3}
	y_1 &  2,4 & 4,3 \\
	\cline{2-3}
	\end{array}
	\]
By Definition~\ref{def:rationality_assumption}, if player 1 is rational assuming that player 2 is both rational and correct, then the pure strategy $s_1=x_1$ would uniquely solve the maximization problem (\ref{eq:maximization_problem2}). Analogously, if player 2 is rational assuming that player 1 is both rational and correct, then the pure strategy $s_1=x_2$ would be the unique solution for player 2. Clearly, player $1$ is not rational, assuming a conjecture of $x_2$.

\printbibliography

\end{document}